\documentclass[conference]{IEEEtran}
\IEEEoverridecommandlockouts
\usepackage[latin1,utf8]{inputenc}
\usepackage{amssymb,amsmath,amsthm, amsfonts}
\usepackage{cite}
\usepackage{algorithmic}
\usepackage{graphicx}
\usepackage{textcomp}
\usepackage{xcolor}
\usepackage{cleveref}
\usepackage{subcaption}
\newtheorem{theorem}{Theorem}
\crefname{section}{§}{§§}
\def\BibTeX{{\rm B\kern-.05em{\sc i\kern-.025em b}\kern-.08em
    T\kern-.1667em\lower.7ex\hbox{E}\kern-.125emX}}
\begin{document}

\title{WISE: Lightweight Intelligent Swarm Attestation Scheme for IoT (The Verifier's Perspective)
}

\author{\IEEEauthorblockN{Mahmoud Ammar}
\IEEEauthorblockA{imec-DistriNet, KU Leuven \\
mahmoud.ammar@cs.kuleuven.be}
\and
\IEEEauthorblockN{Mahdi Washha}
\IEEEauthorblockA{IRIT, Toulouse University \\
mahdi.washha@irit.fr}
\and
\IEEEauthorblockN{Bruno Crispo}
\IEEEauthorblockA{imec-DistriNet, KU Leuven \\
	Trento University, Italy \\
	bruno.crispo@unitn.it}

}

\maketitle

\begin{abstract}
The growing pervasiveness of Internet of Things (IoT) expands the attack surface by connecting more and more  attractive attack targets, i.e. embedded devices, to the Internet. One key component in securing these devices is software integrity checking, which typically attained with \textit{Remote Attestation} (RA). RA is realized as an interactive protocol, whereby a trusted party, \textit{verifier}, verifies the software integrity of a potentially compromised remote device, \textit{prover}. In the vast majority of IoT applications, smart devices operate in swarms, thus triggering the need for efficient swarm attestation schemes.

In this paper, we present WISE, the first intelligent swarm attestation protocol that aims to minimize the communication overhead while preserving an adequate level of security. 
WISE depends on a resource-efficient smart broadcast authentication scheme where devices are organized in fine-grained multi-clusters, and whenever needed, the most likely compromised devices are attested. The candidate devices are selected intelligently  taking into account the attestation history and the diverse characteristics (and constraints) of each device in the swarm. 
We show that WISE is very suitable for resource-constrained embedded devices, highly efficient and scalable in heterogenous IoT networks, and offers an adjustable level of security.      
\end{abstract}



\vspace{-20pt}\section{Introduction}
\label{sec:intro}

Nowadays, Internet-connected embedded devices are increasingly deployed in many aspects of modern life. 
Unfortunately, such devices represent attractive targets for cyber attacks as they are used in a wide range of applications where they process privacy-sensitive information and perform safety-critical tasks. Furthermore, the specialized nature of these devices in terms of limited resources and computing power hinders the realization of strong security mechanisms to prevent cyber attacks. Hence, attacking such devices would  result in very severe consequences, exemplified by Stuxnet \cite{stuxnet}.

Remote attestation (RA) is a well-known security service that enables the detection of malware-infected devices. It allows a trusted entity, denoted as a \textsf{Verifier}, to verify the internal state of remote untrusted IoT device, denoted as a \textsf{Prover}, by measuring its software integrity. Traditional attestation protocols \cite{h2010bootstrapping,spab,smart} consider only a single prover device. However, the diverse applications of IoT are naturally distributed and often embedded in numerous heterogeneous computing devices deployed over wide geographical areas, forming self-organizing mesh networks or swarms. Therefore, applying traditional attestation protocols in such scenarios is inefficient and entails a significant overhead, as the verifier has to run the attestation protocol with each device individually.

\textbf{Problem Statement}. Recently, a number of research proposals have been introduced to address swarm attestation \cite{seda,sana,darpa,lisa,seed,scapi,salad}. To this end, all of these  attestation schemes are efficient and reduce the communication overhead compared to the naive approach of applying many times single prover attestation. Nevertheless, they are sill rigid and not smart enough to adapt the different requirements of the various IoT devices, as whenever attestation takes place, all devices in the swarm are involved without any distinction. 
However, in real scenarios, connected devices in a mesh network have different software configurations and possess heterogenous hardware capabilities. Therefore,  the \textit{equality} approach followed by existing attestation schemes  still incurs substantial overheads.  First, it is harmful to engage safety-critical and time-sensitive  devices in periodic and potentially expensive (in terms of time and energy) RA, especially when the attestation history shows that they have never been compromised or rarely so. Second, efficiency as a feature of scalable attestation is maintained through many factors, and one of them is the aggregation of attestation reports to reduce the message size. Larger the network, the bigger is the size of the aggregate. 
Many IoT devices  (such as IETF Class-1 \cite{RFC7228}) have very limited memories that are not sufficient enough to authenticate and store a big-sized aggregate before propagating it till reaching the verifier. Third, the vast majority of available attestation schemes assume time synchronization which is difficult to handle in heterogenous environments and incurs an extra overhead. 
Last, the nature of some devices (due to their locations, time-sensitivity, tasks, being previously compromised,  etc.) requires to be more frequently attested than others and vice versa. 


\textbf{Contribution.} 
We present WISE, a lightweight, intelligent, and scalable attestation protocol that is very suitable for static heterogenous IoT networks, aiming for minimizing the communication overhead and thus power consumption. In WISE,  the attestation of the entire  swarm is distributed over a  large time window, where a subset of devices is attested  every time interval. The time window can be variable for each device or subset of devices as it is adjusted dynamically at the verifier side.
WISE is smart in the sense that the verifier outputs a candidate subset of devices to be attested after performing online learning from the history of previous attestation periods depending on Hidden Markov Model (HMM) \cite{hmm}, and taking into account the various individual needs and characteristics of connected devices in the swarm as conditions (or features). 

\textbf{Paper outline.} The remainder of this paper is organized as follows. \cref{sec:related} reviews the related work. A system model is presented in~\cref{sec:model}. \cref{sec:wise} describes WISE in details, whereas  \cref{sec:wisef} highlights on the analytical model of the intelligence part of WISE. The security analysis of WISE is discussed in \cref{sec:sec}. Implementation details and evaluations are reported in \cref{sec:impeval}.  \cref{sec:conclusion} concludes the study.

\section{Related Work}
\label{sec:related}

\textbf{Single Device Attestation.} Prior work in single-prover RA can be divided into three main categories: hardware-based techniques~\cite{h2010bootstrapping} that rely on a built-in secure hardware (e.g. Trusted Platform Module (TPM)) in the prover device, software-based techniques~\cite{spab} that do not require any hardware and  depend only on some strict assumptions (e.g. accurate time measurements), and hybrid approaches~\cite{smart} that build atop minimal hardware modifications to the existing low-end IoT platforms \cite{min2014}. Abstractly speaking, each of these categories has pros and cons. Nevertheless, all of them target a single-device attestation. Accordingly, they are not efficient for swarm attestation. 

\textbf{Swarm Attestation.} The recently proposed attestation schemes \cite{seda, sana, darpa,lisa,seed,scapi,salad,slimiot} have addressed the attestation  problem at a large scale where a significant number of devices have to be attested efficiently and securely. These proposals differ from each other in various ways:  (i) the attestation protocol implemented (e.g. interactive \cite{seda, sana, darpa,lisa,scapi,salad,slimiot} or non-interactive \cite{seed}), (ii) the methodology followed in attesting the entire swarm (e.g. tree-based approaches \cite{seda,sana,lisa,seed,scapi,darpa,slimiot} or distributed ones \cite{salad}), (iii) the adversary model considered (e.g. remote-only attacker \cite{seda,sana,lisa,seed}, or remote and physical attackers \cite{darpa,scapi,salad,slimiot}), (iv) the type of the mesh network (e.g. highly dynamic one \cite{salad}, or fairly dynamic and static topology \cite{seda,sana,lisa,seed,scapi,darpa,slimiot}), (v) the consideration of the mesh network (e.g. neighbour discovery, heterogeneity, broadcasting, unicasting, etc.), and (vi) the cryptography employed (symmetric key \cite{lisa,scapi,slimiot}, public key \cite{sana,salad}, or both \cite{seda,darpa}). 

To the best of our knowledge, all existing swarm attestation protocols attest all devices in the swarm whenever the attestation routine is conducted. Hence, we propose WISE, a smart and flexible swarm attestation protocol.

\section{System Model}
\label{sec:model}

\textbf{Overview.} We consider a mesh network with various heterogenous low-end embedded devices. The network should be static or quasi-static~\footnote{Devices only move within the same covering communication range as when they are static.} which we argue that this is the case of the majority of IoT application domains (e.g. smart buildings, etc.). The verifier should know in advance the exact network topology. 
Henceforth, we refer to the verifier as $\upsilon$ and the prover as $\rho$.

\textbf{Device Requirements.} We assume that each device $D_i$ in a swarm $S$ satisfies the minimum hardware properties required for secure RA~\cite{min2014} which are ROM and memory protection unit (MPU). The ROM is needed to store the attestation routine and cryptographic keys, whereas the MPU is required to enforce access control over secret data. 
Also, we assume that $\upsilon$ is a powerful device, i.e. Raspberry PI class or higher. 

\textbf{Adversary Model.} We consider a software-only adversary, $adv$, who has the ability to perform remote  attacks (Dolev-Yao model \cite{dolev}). Thus, $adv$ has full access to the network and can either perform passive (e.g. eavesdrop on communication, etc.) or active (e.g. inject a malware, etc.) attacks.  In line with the majority of existing swarm attestation techniques, We rule out all kinds of DoS and  physical attacks.

\section{WISE: Protocol Description}
\label{sec:wise}

\textbf{Overview.} WISE is a smart scalable attestation protocol which consists of two different phases. The first phase (\cref{subsec:init}) takes place once before the deployment where all devices are initialized by a trusted party (in our case, $\upsilon$) with some public and private data, and distributed logically over different clusters. In the second phase (\cref{subsec:attest}), whenever needed, $\upsilon$ performs the collective attestation over a subset of devices (instead of the entire swarm)  that are wisely selected depending on the attestation history of every device and taking into account the individual characteristics of these devices. 
\begin{figure}[b]
	\vspace{-10pt}
	\centering
	\includegraphics[width=.5\columnwidth]{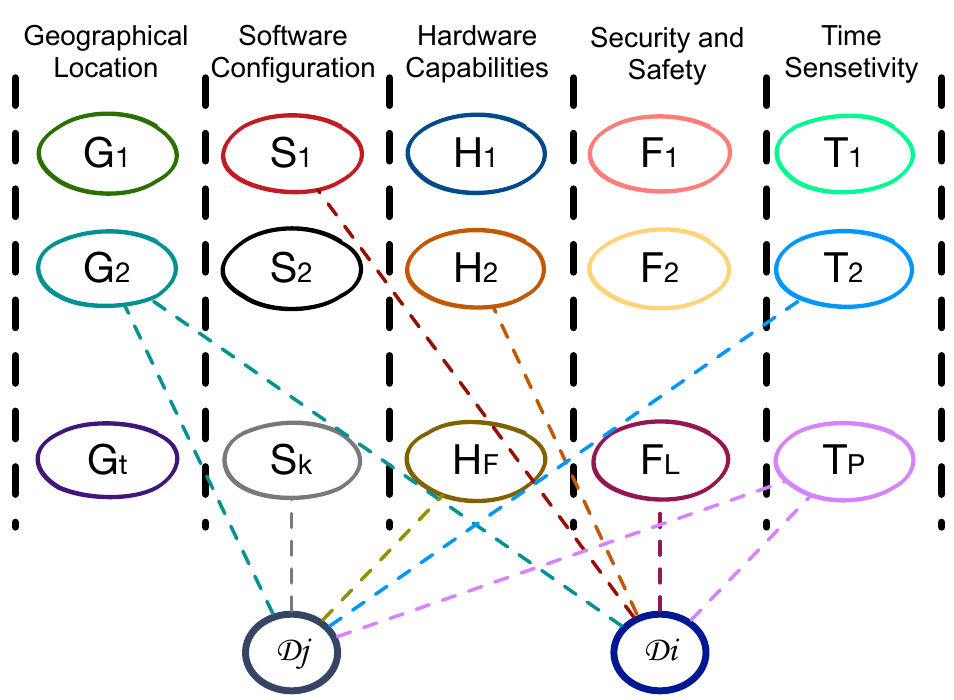}
	\caption{An example of classifying two devices into clusters}
	\label{fig:cluster}
\end{figure}
\begin{figure*}[t]
	\centering
	\includegraphics[width=0.9\textwidth, height=0.9\columnwidth]{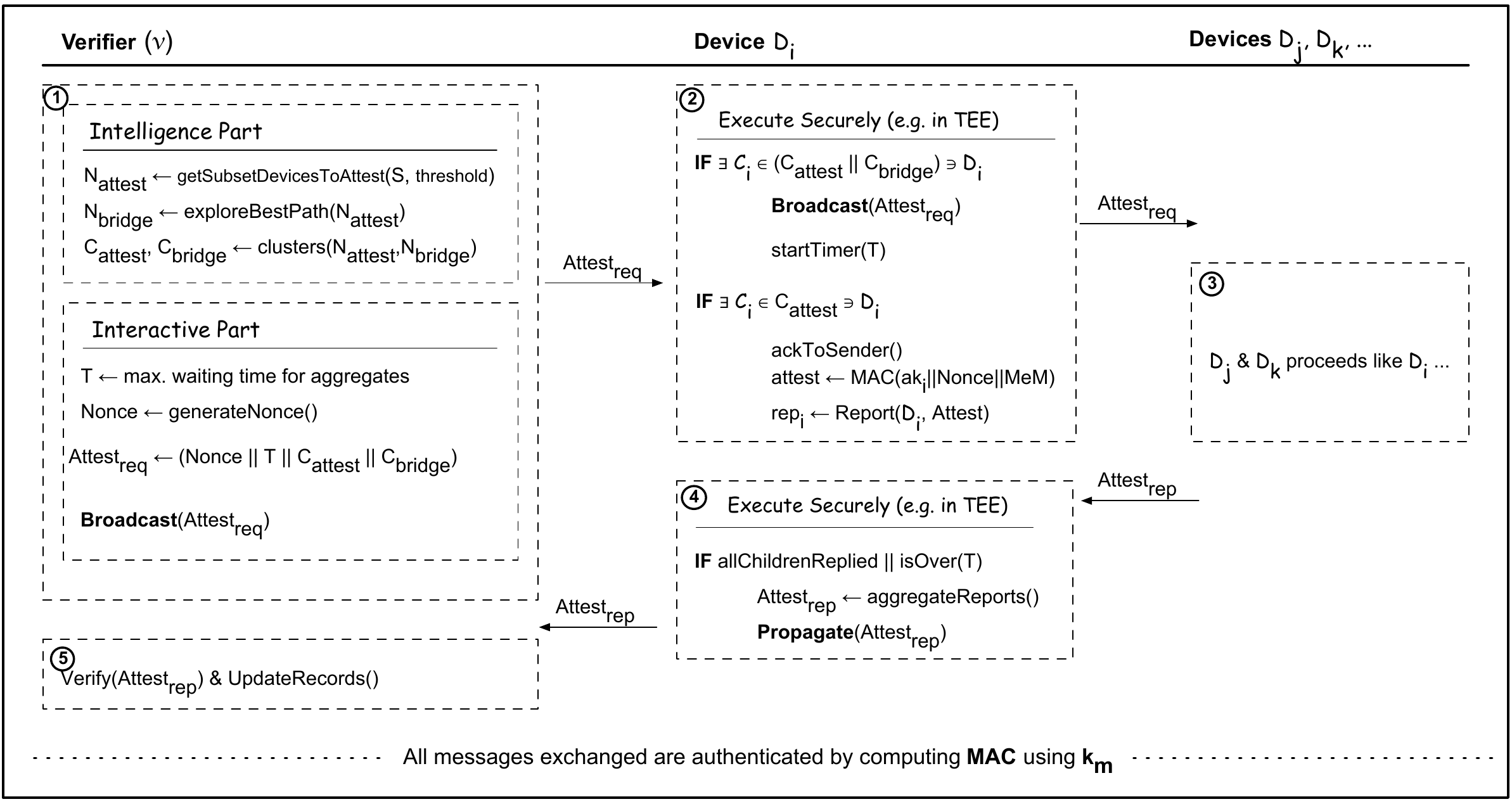}
	\caption{An overview of the attestation phase in WISE}
	\label{fig:attest}
	\vspace{-15pt}
\end{figure*}
\subsection{Initialization Phase}
\label{subsec:init}
$\upsilon$ initializes each $\rho$ with two device-specific keys for mutual authentication ($dk_i$) and attestation ($ak_i$) purposes, and one master key that is shared among all devices at the swarm-level ($k_m$). Also, each $\rho$ stores its own unique ID ($D_i$), and the message authentication code (MAC) of the correct software state ($M_s$) of its memory (MeM), computed using $ak_i$.

Based on the recent taxonomies of addressing security challenges and IoT attacks \cite{tax1,tax2}, we identify five main categories to classify IoT devices, where each category has different clusters and every cluster has a unique ID ($c_i$):
\begin{itemize}
	\item {Geographical Location:} $\rho$'s are divided into clusters based on their geographical locations.
	\item {Software Configurations:} $\rho$'s are grouped into clusters according to partial or full similarities in their software.
	 \item{Hardware Capabilities:} $\rho$'s are grouped into clusters according to the type of hardware they possess.
	\item{Security \& Safety:} its an application-dependent feature which classifies $\rho$'s according to how critical they are to the application domain considering a combination of the task they perform and the security mechanism employed by the hosting hardware.
	\item{Time Sensitivity:} $\rho$'s are classified into clusters according to their time constraints and how critical is it  to deviate them frequently from doing their tasks. 
\end{itemize}

Number of clusters in each category is an application-dependent value. Each $\rho$ should \textit{virtually} belong to one cluster of each category by storing its identifier ($c_i$) in its memory as shown in Figure \ref{fig:cluster}. 
The goal of the multi-clustering technique is twofold; First, it  reduces the communication overhead in terms of number and sizes of exchanged packets, as explained in \cref{sec:wisef}.  Second, over time, $\upsilon$ can learn from the attestation history the behavior of $adv$ by realizing the clusters that have frequently compromised devices, thus being able to patch common vulnerabilities. 
Please note that the selection criteria of categories and clusters is flexible as more or less categories can be considered depending on the application domain.
\subsection{Attestation Phase}
\label{subsec:attest}
Initially, $\upsilon$ initiates the attestation routine for all devices as there is still no history and thus all $\rho$'s in the swarm are attested. Depending on  authenticated broadcast messages, $\upsilon$ propagates an attestation request ($Attest_{req}$), containing a $nonce$ to avoid replay attacks, to all $\rho$'s in the vicinity, which is then further propagated in the network via broadcasting too, thus forming a virtual spanning tree.
$Attest_{req}$ contains also the maximum time ($T$), through which,  each $\rho$ has to send a valid reply, where $T$  is adjusted w.r.t the slowest  $\rho$ included in that attestation. The calculation of various time values takes place once all classes of devices are in a healthy state at deployment time to avoid delay-related attacks (e.g. downgrade or black hole attacks) \cite{black}. Therefore,  every receiver $\rho$ acknowledges  receiving $Attest_{req}$ to the sender device, thus determining the number of replies that should be received. Each $\rho$ answers $Attest_{req}$ by producing an attestation report ($rep_i$). $rep_i$ consists of the device identifier ($D_i$) in addition to an attest value reflecting its state. 
All attestation reports are propagated back after being aggregated at every intermediate $\rho$ using a secure aggregation scheme \cite{aggregate08} till reaching $\upsilon$. 

After running the attestation routine few times \footnote{Number of times is an application-dependent value that determines the level of confidence. In our case, we consider the worst scenario when the attacker behaves randomly. Therefore, a high number of times, e.g. thousands of times, does not really enhance the accuracy. In the experiment section later on, we adjust this number to be only 100 times.}, $\upsilon$ starts behaving smartly by distributing the attestation over a dynamically variable time window, where each time the most likely compromised devices are only attested (see \textcircled{1} in Figure \ref{fig:attest}). In the selection process of the candidates, $\upsilon$ takes into account the following factors: (i) the attestation history of each $\rho$ and its direct neighbours, (ii) the ratio of compromised devices in every cluster that $\rho$ belongs to, (iii) the maximum time that $\rho$ could stay without being attested (this is specified in the initialization phase), (iv) the degree of importance of $\rho$'s task to be attested more often (this can be inferred from Security \& Safety category), and (v) the time constraints of $\rho$ w.r.t real time tasks that should not be  interrupted (or rarely so). In principle, other criteria can be considered in the selection process of devices and the frequency of performing attestation. The minimum number of candidates (i.e. \textit{threshold}) is an application-dependent value.

In order to keep the size of $Attest_{req}$ small, the attestation target will be at the cluster-level. 
Therefore, upon learning, $\upsilon$ finds out the near-optimal set of clusters ($C_{attest}$) (in terms of number of devices and number of hops) that cover the output candidates in order to minimize the communication overhead (given the network topology). Furthermore, if some of the selected devices are not directly reachable by $\upsilon$ or by other selected devices, $\upsilon$ finds out the minimal set of clusters ($C_{bridge}$) that would participate in the attestation routine as bridges. Their goal is only broadcasting the received $Attest_{req}$ and then propagating the delivered attestation reports. Unless there are no other alternatives to reach some or all devices in  $C_{attest}$, $\upsilon$ avoids including the clusters that contain time-sensitive devices in $C_{bridge}$.


Once both sets of clusters are generated, $\upsilon$ creates $Attest_{req}$ containing them along with the $nonce$ value and $T$. Any receiver $\rho$ re-broadcasts $Attest_{req}$ if it belongs to one of the clusters included in either sets, and computes its attestation report if it is only a member in one of the clusters in $C_{attest}$ (see \textcircled{2}). Finally, each $\rho$ propagates the aggregate of attestation reports if either all children are replied or the maximum waiting time is over (see \textcircled{4}). $\upsilon$ verifies the delivered aggregates and accordingly updates the history records of each $\rho$, and then measures are taken against compromised devices, i.e. \textit{secure erasure} \cite{speed}.

\section{WISE: Analytical Model}
\label{sec:wisef}
In this section, we describe the formal mathematical model of the intelligence part of WISE, whereas in the next section (\cref{sec:sec}), we analyze the security of our protocol.
\subsection{Notations and Definitions}
\label{nd}
Let $G(S,L)$ be a graph of an IoT mesh network where $S$ (the swarm) is a finite set of devices (nodes) represented as $S = \{D_1, D_2, \dots, D_n\}$, and $L$ is the set of direct links between nodes. We define $k$ uncorrelated categories as $Cats=\{C_{1},C_{2},\dots,C_{k}\}$ where $C_{\bullet}$ represents the category name (i.e. geographical). Each $C_{\bullet}$ contains a set of clusters, defined as $C_{\bullet}=\{c_{\bullet1},c_{\bullet2}, \dots, c_{\bullet m}\}$ where $m$ is the number of clusters. Each $D_i$ $\in$ $S$ should belong to only one $c_{\bullet i}$ in every $C_{\bullet}$. In nutshell, the following properties must be satisfied:
\begin{enumerate}
	\item\textbf{ Non-overlapping:} For each $C_{\bullet}$, $\bigcap \limits_{1 \leq i \leq |C_{\bullet}|} c_{\bullet i} = \emptyset$. 
	\item\textbf{ Network Coverage:} For each $C_{\bullet}$, $\bigcup \limits_{1 \leq i \leq |C_{\bullet}|} c_{\bullet i} = S$. 
\item \textbf{Uniformity:}  $\forall D_{\bullet}\in S$, $\sum_{i=1}^{k}  | \{D_{\bullet}\} \cap  \bigcup \limits_{1 \leq j \leq |C_{i}|} c_{ij}  |= k$.

\end{enumerate} 

Given the network topology, each $D_i \in S$ is represented as a $<key, value>$ map, where $key$ is the unique $D_i$ and $value$ is an 8-tuple of metadata represented as follows: $\langle D_{clusters}, T_{max}, D_{history}, D_{attest,t},T_{last}, D_{neighbors},$ 
$D_{paths}, D_{degree}\rangle$. Each attribute is defined as follows:

\begin{itemize}
	\item $D_{clusters}:$ the $k$ clusters where $D_i$ is a member in.
	\item $T_{max}\in \mathbb{Z}^+:$ the maximum time (in minutes) that $D_i$ can stay without attestation. 
	
	\item $D_{history}:$ the attestation history of $D_i$, represented as an ordered finite set $\{a_{0},a_{1},\dots,a_{t-1}\}$, where $a_{\bullet} \in \{0,1,x\}$ as 0 means that $D_i$ was compromised, 1 refers to a healthy state, and $x$ means that $D_i$ has not been attested in the considered iteration.  
	
	\item $D_{attest,t}  \in \{0,1,x\}:$ the predicted attestation response in the current iteration ($t$).
	 
	\item $T_{last}\in \mathbb{Z}^+:$ the last time (in minutes) that $D_i$ has been attested. 
	
	\item $D_{neighbors} \subset S:$ the set of devices that are directly connected with $D_{i}$, represented as $D_{neighbors}=\{D_j|l_{ij}\in L,D_{j}\in S\}$ where $l_{ij}$ is an undirected link.
	
	\item $D_{paths}:$ an ordered set of the top $m$ shortest paths that $\upsilon$ can follow to reach $D_i$. Each path is represented as a set of devices from source ($\upsilon$) to destination ($D_i$).
	
	\item $D_{degree}\in [1,|C_{k}|]$: it represents the degree of sensitivity w.r.t  real time tasks where 1 means extremely sensitive. 
\end{itemize}
\vspace{-3pt}
\subsection{Problem Formalization}
Given the $<key, value>$ map of devices, the intelligence part of WISE (see \textcircled{1} in Figure \ref{fig:attest}) can be represented as multi-objective optimization problem.

\subsubsection{Candidates Selection}
\label{sss:cs}
$\upsilon$ leverages the attestation history of each $D_i$ as a heuristic information for computing the probability of being uncompromised in the next attestation iteration. Considering the calculated probabilities of devices along with their time constraints, the candidates selection problem is formalized as follows: \vspace{-5pt}\begin{equation}
\label{eq1}
\resizebox{1\hsize}{!}{$
\begin{aligned}
& \underset{\alpha_{1},\dots, \alpha_{|S|} }{\text{minimize}}
\ \sum_{j=1}^{|S|} \alpha_{j}*Pr(D_{j}.D_{attest,t}=1|O_t) \ \ \text{s.t. }\ \sum_{i=1}^{|S|} \frac{\alpha_{i}}{|S|} \geq Th_{cov}, \\
& Min\{(1-\alpha_{i})*(D_{i}.T_{last}+D_{i}.T_{max}-T_{current}):
1\leq i \leq |S| \}>0 
\end{aligned} $}
\end{equation}

where  $|S|$ is the number of devices in the swarm, $\alpha \in \{0,1\}$ is a binary free variable at which the optimization step is performed for either excluding (0) or including (1) $D_j$ in the attestation, $Pr(D_{j}.D_{attest,t}=1|O_t) $ represents the probability of the $j^{th}$ device being uncompromised in the current attestation iteration (once the interactive part takes place) given $O_t$ as a set of features (observations), $T_{current}$ is
the time (in minutes) of running the interactive part of the current iteration. $Th_{cov}$ is a constraint in the objective function which represents the minimum ratio of devices to be covered. The second constraint guarantees attesting all devices within the predefined time window. The output is a set of the most likely compromised devices that is formally defined as
 $N_{attest}=\{D_{i}: D_{i}\in S, \alpha_{i}=1\}$.
 
 \subsubsection{Best Route Exploration (gap bridging)}
 WISE depends on an authenticated broadcast technique to emit $Attest_{req}$. Getting $Attest_{req}$ delivered to all candidates selected in \cref{sss:cs} should incur a minimal communication and runtime overhead. Therefore, a minimum number of devices should be involved in delivering $Attest_{req}$ to the candidates 
 and then propagating their aggregated reports to $\upsilon$. Unless there are no other alternatives, time-sensitive devices should not be involved and deviated from doing their main tasks. The gap bridging problem is formalized as follows:
 \vspace{.7mm}
 \begin{equation}
 \resizebox{.9\hsize}{!}{$
 	\centering
 \begin{aligned}
 & \underset{\beta_{1},  \dots  \beta_{|P_{comb}|} }{\text{minimize}}
 \ \sum_{ \langle P_1,\dots,P_{|N_{attest}|}\rangle_{j}\in P_{comb} } \beta_{j}* \sum_{D_{\bullet}\in U_{D_i}}  D_{\bullet}.D_{degree} \\ 
 &\text{s.t. }\ \ \sum_{i=1}^{|P_{comb}|} \beta_{i} =1
 \end{aligned}$}
 \end{equation} where $P_{comb}=\prod_{a=1}^{|N_{attest}|}D_{a}.D_{paths}$ is the Cartesian product of the set of paths of each candidate device, $|P_{comb}|$ is the cardinality (length) of path combinations set, $\langle P_1,\dots,P_{|N_{attest}|}\rangle_{j}$ is the $j^{th}$ tuple element in $P_{comb}$ which corresponds to best (shortest) paths combination among $|P_{comb}|$ different possibilities, $\beta_{j}\in \{0,1\}$ is a binary free optimization variable for the $j^{th}$ paths combination, and   $U_{D_i}=\bigcup_{i=1}^{i=|N_{attest}|} P_{i}$ represents a set of nodes resulted by unifying a combination of shortest paths associated with the $j^{th}$ tuple $\langle P_1,\dots,P_{|N_{attest}|}\rangle_{j}$. The imposed constraint is to ensure the selection of a single combination of paths. The output is a set of nodes, defined as $N_{bridge}=\{D_{\bullet}: \langle P_1,\dots,P_{|N_{attest}|}\rangle_{j}\in P_{comb}, D_{\bullet}\in \bigcup_{i=1}^{i=|N_{attest}|} P_{i},  \beta_{j}=1 \}$. 
      
 \subsubsection{Clusters Selection}
 adding the IDs of candidate devices (along with bridge devices) to $Attest_{req}$ hinders the scalability as the message size will be big. This also sharply increases the communication overhead as well as the processing time at $\rho$ side.  We overcome this limitation by targeting the cluster-level. In such case, $Attest_{req}$ will include the identifiers of the minimal set of clusters that cover the selected devices w.r.t. the communication overhead. This is formalized as follows:
 \vspace{-10pt}
 \begin{equation}\begin{aligned}
 & \underset{\gamma_{11}, \gamma_{12}, \dots  \gamma_{|C_{k}|m} }{\text{minimize}}
 \ \sum_{j=1}^{k} \sum_{i=1}^{|C_{j}|} \gamma_{ji}*ComOverhead(c_{ji}) \ \ \text{s.t. }\ \\
 & \forall D_{\bullet}\in N_{attest} \cup N_{bridge}\  \sum_{j=1}^{k} \sum_{i=1}^{|C_{j}|}  \gamma_{ji}*|c_{ji}\cap \{D_{\bullet}\}|\geq1
 \end{aligned}
 \end{equation}
 
 where $k$ is the number of categories, $|C_{j}|$ is the number of clusters in the $j^{th}$ category, $c_{ji}$ represents the $i^{th}$ cluster in the $j^{th}$ category, $ComOverhead(\bullet)$ is a function that computes the communication overhead in terms of number of hops, $\gamma_{ji}\in \{0,1\}$ is a binary free variable at which the optimization process is performed so that $1$ value means that the cluster $c_{ji}$ is selected. The constraint imposes selecting the most relevant clusters to the unified set of $N_{attest} \cup N_{bridge}$. The output is two sets of clusters represented as
 $C_{attest}=\{c_{ji}: c_{ji} \in C_{j}, j\in\{1,\dots,k\}, i\in\{1,\dots,|C_{j}|\}, \gamma_{ji}=1, c_{ji} \cap N_{attest}\neq\emptyset \}$, and
  $C_{bridge}=\{c_{ji}: c_{ji} \in C_{j}, j\in\{1,\dots,k\}, i\in\{1,\dots,|C_{j}|\}, \gamma_{ji}=1, c_{ji} \cap N_{attest}=\emptyset \}$  respectively.
\newline  
\newline
  
  \subsection{Probabilistic Model \protect\footnote {We identify three heuristic approaches as solvers to provide near-optimal solutions for each of the formalized optimization problems. Due to the limited space, we omit the description of these approaches. To get the pseudo code of these solvers (or even the complete source code of the intelligence part), please contact the first author.}}
  
To compute the probability component, $Pr(D_{\bullet}.D_{attest,t}|O_t)$,  in the candidates selection problem, we use HMM for two main reasons. First, it computes the probability distribution of discrete hidden (unknown) states (i.e. compromised or uncompromised) over sequences (time-series) of observations (i.e. number of compromised neighbours). Accordingly, it is suitable to deal with the worst case scenario when the attacker works randomly. Seconds, in spite of assuming a powerful verifier, it is compact and very fast to use, compared to other techniques (i.e. Recurrent Neural Networks) which are very time-consuming (iterative-based) and require a very big memory to store the $<key, value>$ map when considering large mesh networks.

\textbf{HMM notations}. At any iteration $t$, two hidden states are identified: compromised ($D_{attest,t} = 0$) and uncompromised ($D_{attest,t} = 1$). Considering the $<key, value>$ map, a vector of observations $\textbf{O}_t$ is formed for each $D_{\bullet}$ by selecting $M$ random independent variables (features). Given a sequence of observations $\{\textbf{O}_{1},\dots,\textbf{O}_{t}\}$, learning $D_{attest,t}$ requires computing the following HMM parameters $\theta=\{\pi,H,b_{0,f_1}, \dots b_{0,f_M},b_{1,f_1},$ $\dots b_{1,f_M} \}$:
\begin{itemize}
	\item {The state transition matrix ($H$):} it contains the probabilities of  transitioning from a state to the other. For instance, $h_{0,1}$ is the probability of $D_{\bullet}$ to be healthy given that $D_{\bullet}$'s previous state  is compromised (i.e. $Pr(D_{attest,t}=1|D_{attest,t-1}=0)$).
	\item {The initial state probability distribution ($\pi$):} it initializes the probability distribution over $|Q|$ values, where $Q$ is the set of hidden states and $|Q|$ is its cardinality.
	\item{The emission (observation) probability distributions (\{$b_{1,f_1}, \dots b_{1,f_M},b_{0,f_1}, \dots b_{0,f_M}  $\}): } it requires $|Q|\times M$ distributions for inferring values. For instance, $b_{0,f_1}$ is the probability distribution of observing the value of feature $f_1$ (i.e. number of compromised devices in a joint cluster)  when $D_{\bullet}$ is predicted to be compromised.
\end{itemize}

\textbf{Probability Computation.} By leveraging the $D_{\bullet}$'s attestation history set ($D_{history}$), we estimate the HMM parameters, $\theta$, using the Maximum Likelihood Estimation (MLE) method  for a given set of observations. More precisely, the MLE  finds the values that maximize the likelihood of generating the given observations. This is formalized as:\begin{equation}\label{eqa4}\begin{aligned}
& \theta_{mle}= \underset{\theta \in \Theta }{\text{arg max}}\space Pr(L|\theta) \ \ \ 
\text{s.t.} \ \ \ \sum_{i\in Q}\pi_{i}=1, \\
&   \prod_{j\in Q} \sum_{i\in Q}h_{i,j}=1,\ \  \prod_{i\in Q}\prod_{v\in \{1,\dots,M\}}  \int_{\mathbb{R}} b_{i,v}(x) dx =1 
\end{aligned}
\end{equation}
 
where $\Theta$ is the search space of HMM parameters,  $L=\{(x_{1},y_{1}),..., (x_{m},y_{m})\}$ is a set of $m$ distinct groups of paired sets of observations and $D_{\bullet}$'s attestation responses, $(x_{\bullet},y_{\bullet})=(\{\textbf{O}_{t_1},..., \textbf{O}_{t_2}\},\{a_{t_1},...,a_{t_2}\})$, $1\leq t_1<t_2$, $x_{\bullet}$ is a sequence of observation (feature) vectors, and $y_{\bullet}$ is a sequence of true hidden states from iteration $t_1$ to $t_2$.  The first two constraints should equal 1 according to the HMM. The last constraint is associated to $|Q|\times M$ emission probability distributions where the area under each distribution should equal 1 and accordingly the multiplication of the $|Q|\times M$ distributions area  equals 1.

The probability of generating any sequence of observations can be written using the HMM parameters $\theta$ as follows:\vspace{-5pt}   \begin{multline}
Pr((x_{\bullet},y_{\bullet})|\theta )= \prod_{i\in Q } \pi_{i}^{f(i,y_\bullet)} \ \ *\ \ \prod_{i\in Q } \prod_{j\in Q }  h_{i,j}^{f(i,j,y_\bullet)}\\ * \prod_{i\in Q}\prod_{\textbf{O}_t \in x_{\bullet} }   \prod_{o_{t,v}\in\textbf{O}_t  }    b_{i,v}(o_{t,v})^{f(i,v,o_{t,v},x_{\bullet},y_{\bullet})}
\end{multline} 
where $f(i,y_{\bullet})$ is an integer $\in [0,1]$  representing the occurrence of the state $i$ given a sequence $y_{\bullet}$, $f(i,j,y_{\bullet})$ is the number of occurrences of having a transition from the state $j$ to $i$ in the given true states sequence $y_{\bullet}$, and $f(i,v,o_{t,v},x_{\bullet},y_{\bullet})$ is the number of times of observing the value $o_{t,v}$ of feature $v$ in the given sequence of observations $x_{\bullet}$ and true hidden states $y_{\bullet}$ when the HMM process is in the state $i$.  

 The values of $\theta$ parameters can be obtained by assuming that for each sequence of the training $m$ groups, the probability of generating  $L$  equals to the probability multiplication of generating each  sequence individually, formalized as:\vspace{-10pt}\begin{multline}
Pr(L|\theta)=\ \ \ \prod_{l=1}^m Pr((x_{l},y_{l})|\theta)=\ \ \ \prod_{l=1}^m \Bigg( \prod_{i\in Q } \pi_{i}^{f(i,y_\bullet)}  \\ * \prod_{i\in Q } \prod_{j\in Q }  h_{i,j}^{f(i,j,y_\bullet)}  \prod_{i\in Q}\prod_{\textbf{O}_t \in x_{\bullet} }   \prod_{o_{t,v}\in\textbf{O}_t  }    b_{i,v}(o_{t,v})^{f(i,v,o_{t,v},x_{\bullet},y_{\bullet})}\Bigg )
\end{multline}
$Pr(L|\theta)$ is a concave function having only one global maximum value. Hence, we compute the partial derivatives of $Pr(L|\theta)$ w.r.t each free parameter in $\theta$. Thus, the values of HMM parameters  that make each of the partial derivatives equals to zero ($\frac{\partial Pr(L|\theta)}{\partial a_{1,1}}=0$, $\frac{\partial Pr(L|\theta)}{\partial \pi_{1}}=0$, $\frac{\partial Pr(L|\theta)}{\partial b_{0,f_1}}=0$, \dots) are computed as follows:\begin{align}
&\pi_{i}=\frac{\sum_{l=1}^m f(i,y_l)}{\sum_{k\in Q}\sum_{l=1}^m f(k,y_l)} \hspace{15pt} h_{i,j}=\frac{\sum_{l=1}^m f(i,j,y_l)}{\sum_{k\in Q}\sum_{l=1}^m f(i,k,y_l)} \hspace{15pt}\\
& b_{i,v}(o)=\frac{\sum_{l=1}^m f(i,v,o,x_{l},y_l)}{{\sum_{l=1}^m \int_{\mathbb{R}} f(i,v,o',x_{l},y_l) do'  }}
\end{align}

The resulted values are normalized and accordingly they implicitly satisfy the constraints in \cref{eqa4}.

\textbf{Inference.} Given the resulted HMM parameters, for any $D_{\bullet}$, the probability of being uncompromised at iteration $t$ is computed as:\begin{align}
\label{eq6}
Pr(D_{\bullet}.D_{attest,t}=1| \textbf{O}_{t})= \sum_{j\in Q}\pi_{j}\ h_{1,j}\prod_{ o_{v} \in  \textbf{O}_{t} } b_{1,v}(o_{v})  
\end{align}


\section{WISE: Security Analysis}
\label{sec:sec}

In principle, the time constraint in \cref{eq1}, that is demonstrated by either a variable time window for each device or a fixed one for all, makes WISE (in spite of being operating in a smart probabilistic way that maintains a minimum communication overhead) at the same deterministic security level  as other existing swarm attestation schemes. For instance, other schemes verify the integrity of the entire swarm each particular period of time (e.g. every 6 hours). During this period (the time between two successive attestation routines), the status of devices is unknown as some of them might be compromised and accordingly they will remain compromised till the next attestation routine is launched. In contrast, WISE distributes the attestation over this long time period (time window), thus minimizing the communication overhead, and speeding up the slow detection of some compromised devices that may result because of the long period of inactivity.
Furthermore, WISE has two other main  advantages. First, it addresses the individual differences and characteristic of devices (e.g. some devices has to be more frequently attested than others and vice versa as shown in Figure \ref{fig:time} where both fixed and variable time window techniques are illustrated). Second, it is robust against a roving malware \footnote{Roving malware is an advanced kind of malware that is always aware of the attestation schedule of the IoT device and thus it is only active between any two successive  attestation routines which deletes itself at the beginning of the attestation to evade detection.}, whereas all other existing schemes are not. In WISE, such kind of malware can not be aware of the attestation schedule because the time window of any device could be variable and represents only the upper bound, by which, the device should be attested at least once, whereas it could be many times as shown for some devices in Figure \ref{fig:fix}, depending on the output of the  intelligence mechanism employed. In nutshell, WISE allows to accommodate different risk levels, rather than adopting a single one as existing schemes, since for big mesh networks, it is unrealistic to assume that all nodes have the same level of exposure to attacks. A higher risk level requires more frequent attestation than a lower one.
\begin{figure}[t!]
	\centering
	\begin{subfigure}[b]{\columnwidth}
		
		\includegraphics[width= \columnwidth]{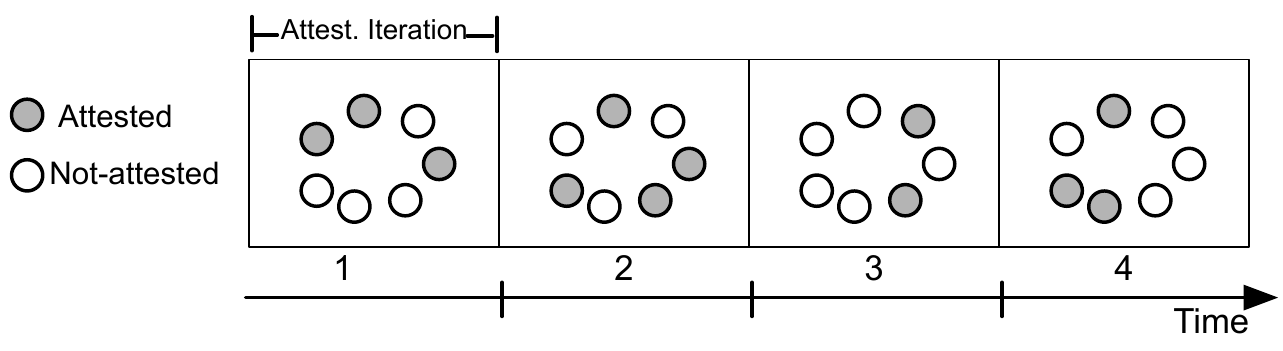}
		\caption{A fixed time window of 4 iterations for a group of devices.}
		\label{fig:fix}
	\end{subfigure}
	\begin{subfigure}[b]{\columnwidth}
		
		\includegraphics[width= \columnwidth]{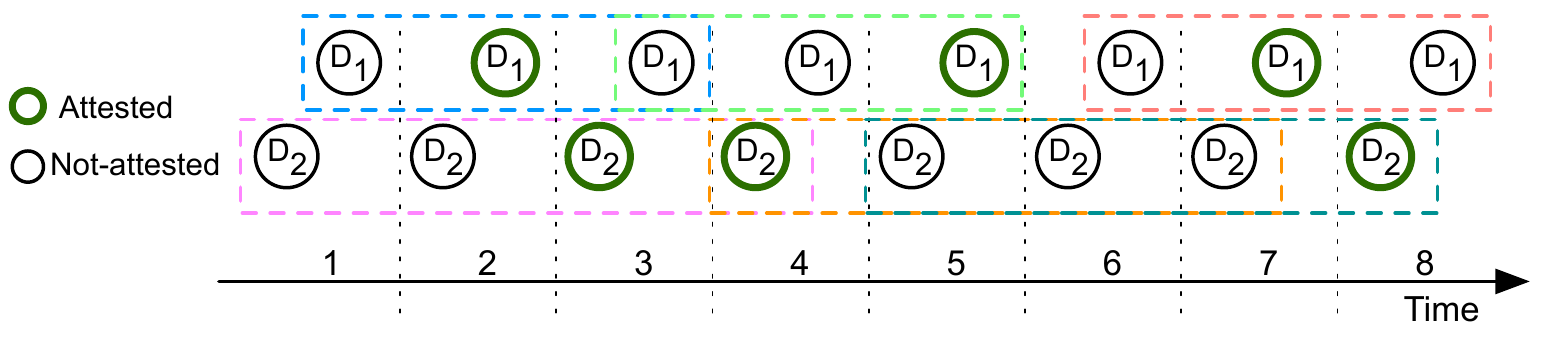}
		\caption{An example of variable time window for two devices, where it is 3 for $D_1$ and 4 for $D_2$.} 
		\label{fig:var}
	\end{subfigure}
	\caption{An overview of the timing technique used in WISE that ensures covering all devices in the attestation within a variable/ fixed time window.}
	\label{fig:time}
	\vspace{-20pt}
\end{figure}

In particular, the communication cost in WISE depends on the tolerated latency of detecting (\textit{possibly}) compromised devices, where in each attestation iteration $A_i$, $m$ out of $n$ devices are attested, and $m$ is an application-dependent value referring to the likelihood of a device to be compromised. WISE avoid starvation cases by identifying a (variable) maximum time window ($W$), through which, each $D_i$ is attested at least once. We elaborate on the security of WISE in two streams: (i) per time window, and (ii) per attestation iteration.

\subsection{Security of WISE per time window.}

This is formalized by the following adversarial experiment $ATT_{adv}^{n, c}(j)$, where $adv$ interacts with $n$ devices, and compromises up to $c$ devices in $S$ according to the  adversarial model presented in \cref{sec:model}, where $0 \leq$ $c$ $\leq$ $n$. Considering that $adv$ is computationally bounded to the capabilities of the devices deployed in $S$, $adv$ interacts with the devices a polynomial  number of times $j$, where $j$ is a security parameter. $\upsilon$ verifies the attestation aggregates received in all $A_i$'s within the time window ($W$) and outputs a decision as 0 or 1. Denoting the decision made  as $A$, $A$ = 1 means that all devices are attested and compromised ones are detected, or $A$ = 0 otherwise. 
\begin{theorem}
	WISE is secure if $Pr$[$A$ = 1 $|$ $ATT_{adv}^{n, c}(j)$ == $A$] is negligible for 0 $<$ $c$ $\leq$ $n$. 
\end{theorem}

\begin{proof}
	Considering \cref{eq1}, $\forall D_{\bullet} \in S$:
	\begin{align}
	\bigcup _{i=1}^{W} \{D_{attest,i}\} \cap \{0,1 \}\neq \emptyset 
	\end{align}
	In other words, each $D_{\bullet}$ in $S$ will be attested $g$ times within $W$ w.r.t $T_{max}$ (see \cref{nd}), where $ 1 \leq g \leq |A_i|$. It is important to mention that this is valid under the assumption that both MAC and aggregation schemes are secure and the attestation routine in the device itself is atomic. 	
\end{proof}

\subsection{Security of WISE per attestation iteration ($A_i$).}
The security of WISE per $A_i$ is a probabilistic value which range from ($1 - \frac{m}{n}$) to $1$, where ($1 - \frac{m}{n}$)  is the worst case when all unselected devices in the current attestation are compromised. Considering the probabilities calculated in \cref{eq6}, we calculate the probability of having a secure swarm (no undetected compromised devices) as: $Pr$[$A_i$ = 1] $\geq$ $Th_{sec}$, where $Th_{sec}$ is a security threshold computed as:  
\begin{align}
Th_{sec} = \frac{m}{n} + \frac{\sum_{D_{\bullet}\notin A_i} Pr(D_{\bullet}.D_{attest,t}=1| \textbf{O}_{t})}{n}
\end{align}

\section{Implementation \& Evaluation}
\label{sec:impeval}

\subsection{Implementation}
\label{subsec:imp}

A prototype of WISE has been implemented over a heterogeneous mesh network consisting of 12 nodes with different hardware and software configurations where each node belongs to Raspberry PI, Arduino Uno, Arduino Due, Arduino Zero, Microship STK 600 with AVR Atmega644P MCU, or MicroPnP IoT platform \cite{mpnp}. Each platform is equipped with IEEE 802.15.4e radio transceiver for wireless communication. SHA-256 is employed as a keyed-hash message authentication code (HMAC) to measure the software integrity during the attestation phase. The precise time consumed by the various cryptographic and interactive operations in each platform is measured in order to simulate large mesh networks.
A Raspberry PI 3 is selected to act as a verifier. In the implementation of HMM, the ratio of compromised neighbours along with the ratio of compromised devices in each corresponding cluster are selected as features (observations) for each device.

\begin{figure*}[t!]
	\centering
	\begin{subfigure}[b]{0.33\textwidth}
		
		\includegraphics[width= \textwidth]{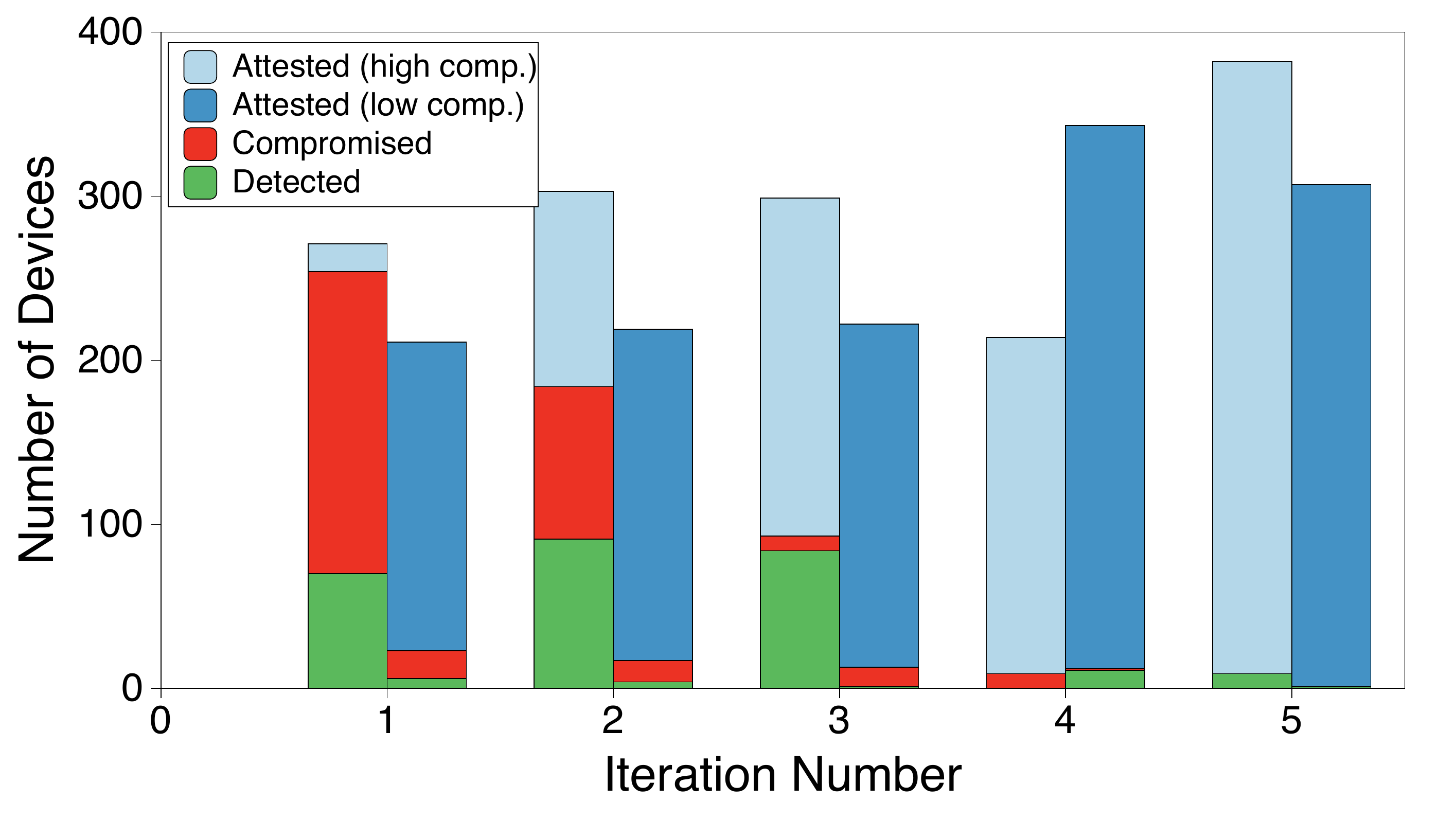}
		\caption{A full scenario of a 1000-node network.}
		\label{fig:sc1000}
	\end{subfigure}
	\begin{subfigure}[b]{0.33\textwidth}
		
		\includegraphics[width= \textwidth]{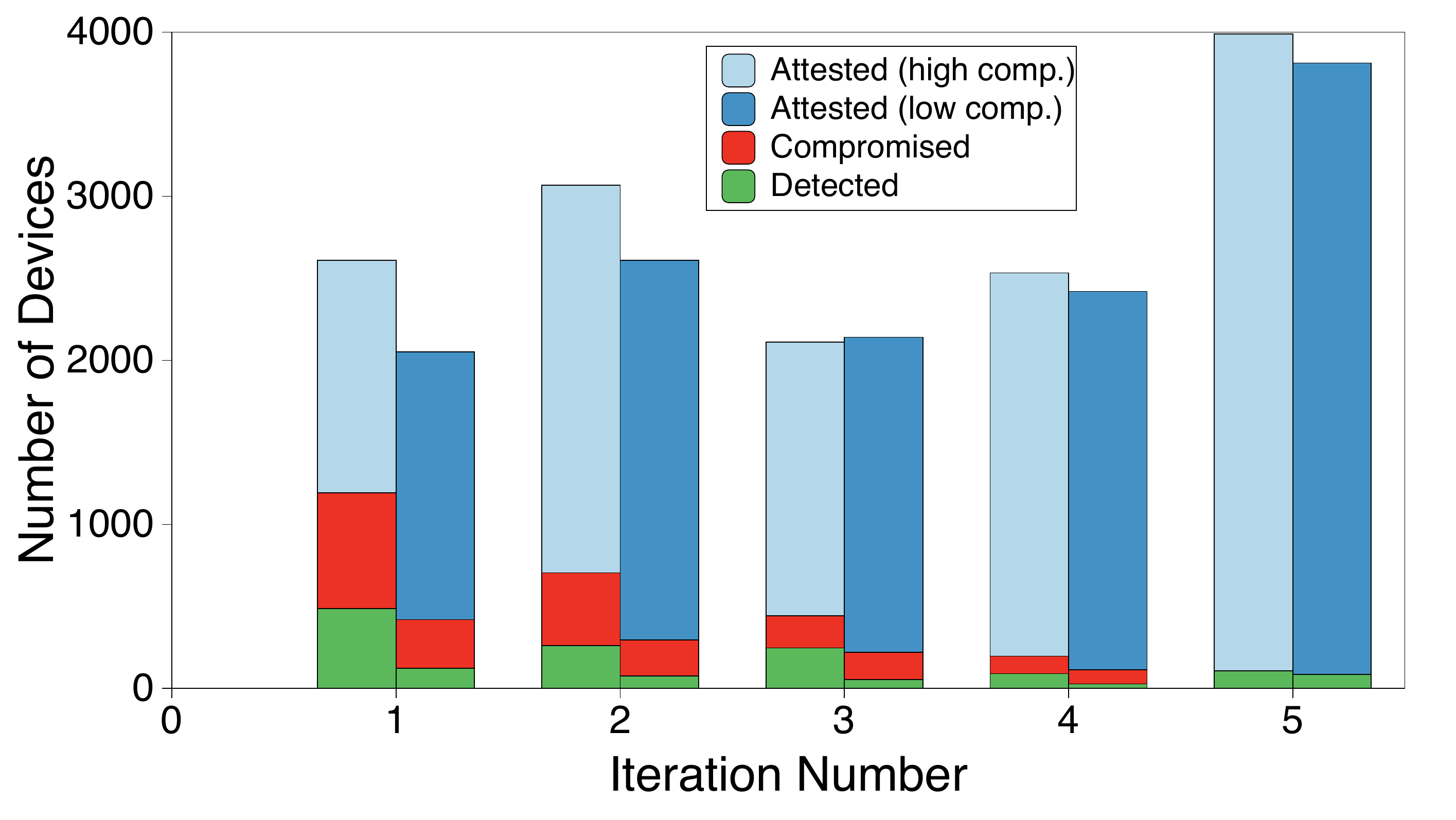}
		\caption{A full scenario of a 10000-node network.} 
		\label{fig:sc10000}
	\end{subfigure}
	\begin{subfigure}[b]{0.32\textwidth}
		
		\includegraphics[width= \textwidth]{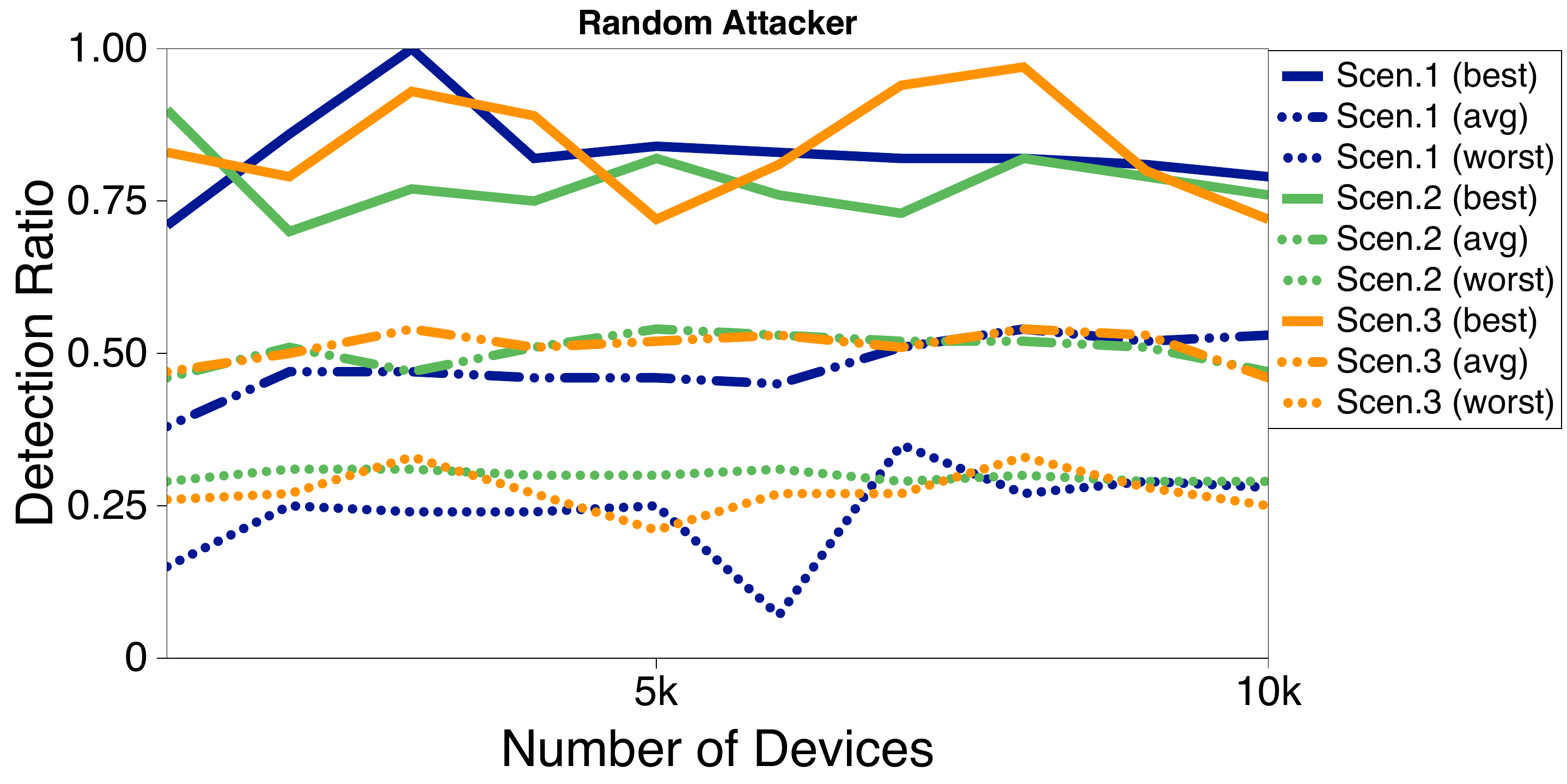}
		\caption{Detection rate with a random attacker.}
		\label{fig:rattack}
	\end{subfigure}
\begin{subfigure}[b]{0.33\textwidth}
	
	\includegraphics[width= \textwidth]{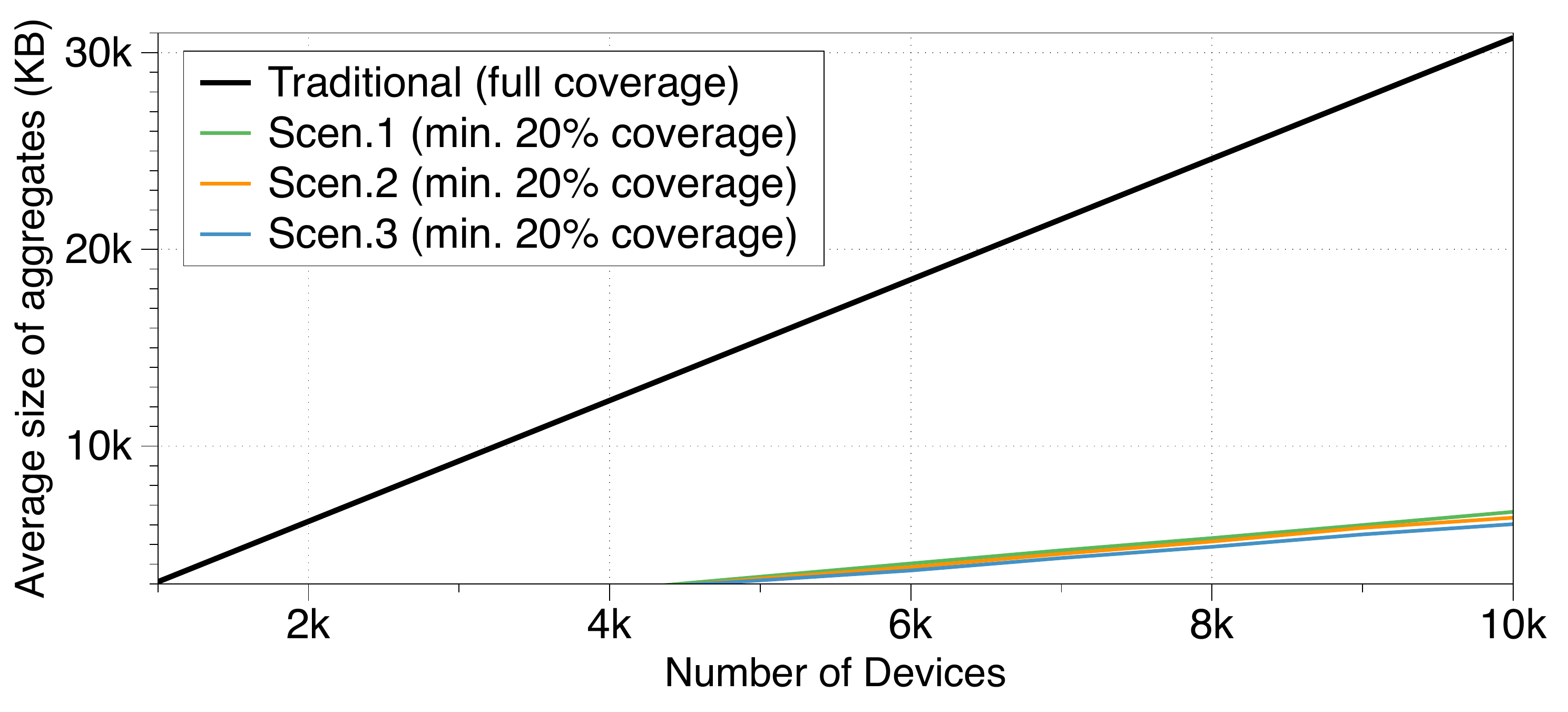}
	\caption{Average size of Aggregated reports.}
	\label{fig:aggregates}
\end{subfigure}
\begin{subfigure}[b]{0.33\textwidth}
	
	\includegraphics[width= \textwidth]{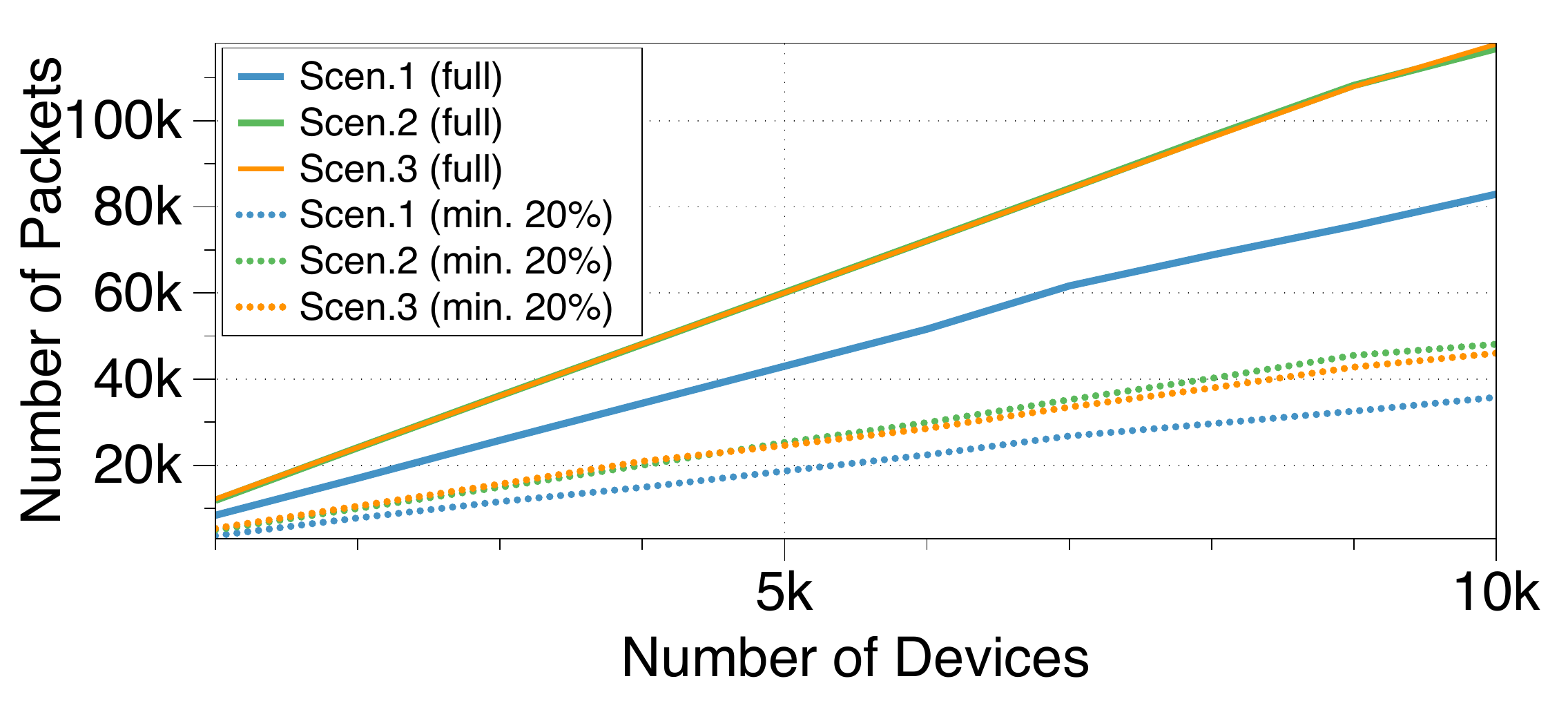}
	\caption{Average number of exchanged Packets.}
	\label{fig:packets}
\end{subfigure}
\begin{subfigure}[b]{0.32\textwidth}
	
	\includegraphics[width= \textwidth]{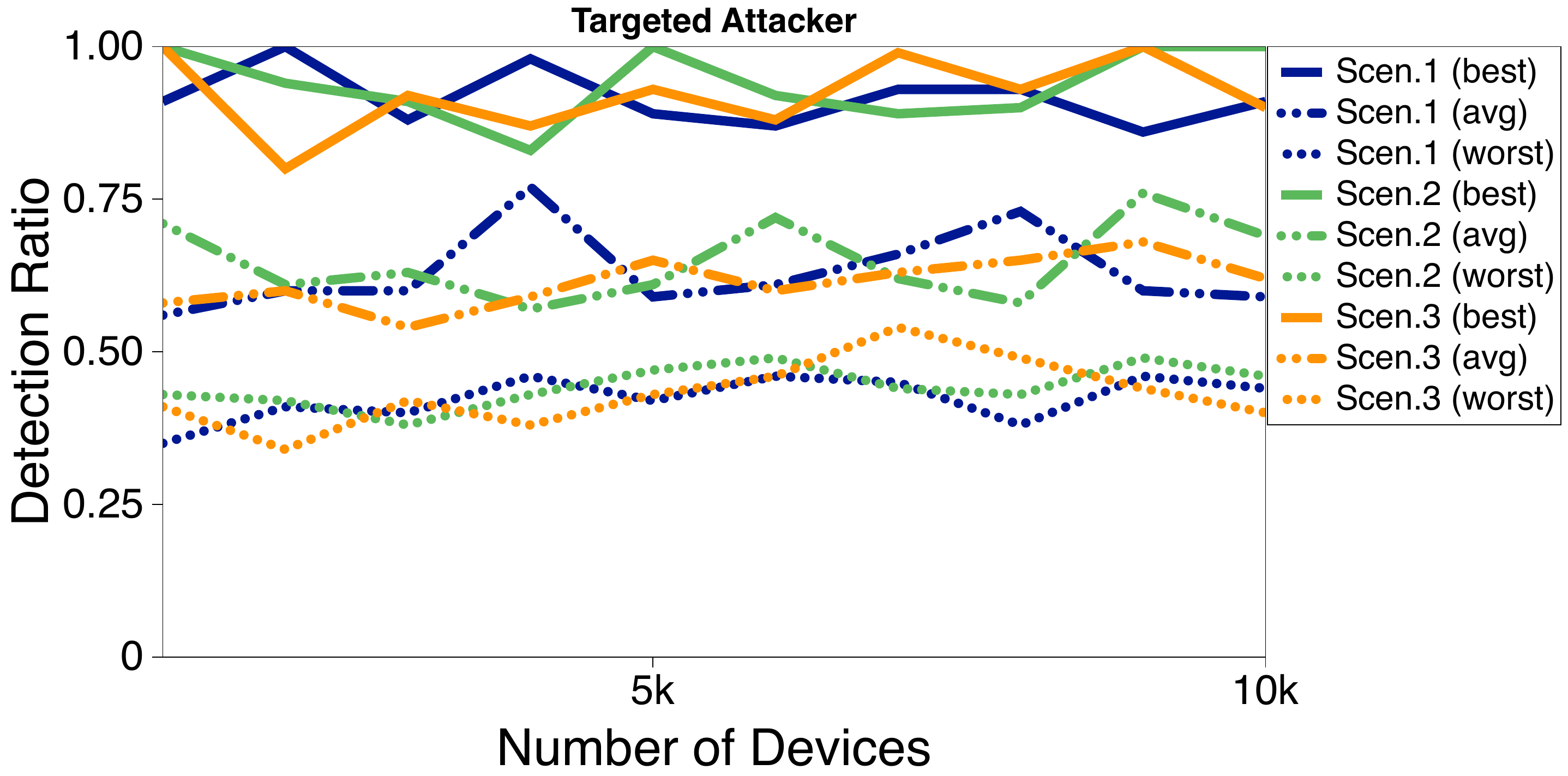}
	\caption{Detection rate with a targeted attacker.}
	\label{fig:tattack}
\end{subfigure}
	\caption{An overview of diverse evaluation factors of WISE}
	\label{figureapp1}
	\vspace{-15pt}
\end{figure*}

\vspace{-3pt}
\subsection{Evaluation}
\label{subsec:eval}

OMNeT++ framework \cite{omnet} is used to simulate large networks with different configurations. The computational and network delays are adjusted according to the experimentally measured values in the prototype implemented. Three different scenarios are simulated, where  each scenario composes of 10 different-sized networks ranging from 1k to 10k devices. The various mesh networks in all scenarios are deployed randomly. In scenario 1, the number of neighbours of each node is set to be a random value from 1 to 5, whereas it is set to be from 1 to 7 in scenario 2 and from 2 to 10 in scenario 3. 

Each contiguous group of devices are selected to be in the same geographical cluster with a fixed size of 30 (mostly), 20, or 10 devices. The devices are deployed randomly over other clusters in other categories. We assigned 6 different clusters in the hardware category and 30 clusters in the software one. The last two categories are composed of 10 different clusters each.

For each network in each scenario, we first run the attestation routine in a uniform state (all devices are attested) 100 times, where each time, we set up an adversarial model to attack randomly a number of devices ranging from 0 to 30\% of the network size. After then, WISE takes place.

In all following figures, we set the minimum coverage of candidate devices per attestation iteration to be 20\% of the network size,  and the number of attacked devices is random in each iteration (0 to 30\% of the network size). 

\textbf{Security and Robustness.}
Figure \ref{fig:sc1000} and \ref{fig:sc10000} represent two examples of two different-sized networks from scenario 1, where in each one,  we show two cases of detecting high and low number of compromised devices. In these figures, the attack took place once before the first iteration.  All compromised devices are detected by the $5^{th}$ iteration (half of the longest time window as we modeled it as a variable discrete event in terms of number of iterations, where for each device, it is set to be a random value from 4 to 10). 

Figure \ref{fig:rattack} shows the worst, average, and best ratio of detected compromised devices (per iteration) after running WISE 30 times in each of the thirty networks in all scenarios. Each time, a different number of devices is randomly compromised. The worst ratio was almost higher than 25\%, whereas the average ratio was around 45\% (with 20\% coverage). Figure \ref{fig:tattack} shows the ratio of detected compromised devices if the attacker was targeting a specific class of devices (in our case, we modeled it to attack Arduino-zero  cluster). Results show that the detection ratio with a targeted attacker is higher than the random one where in many iterations it was 100\% at best, with an average of about 62\%. 

\textbf{Scalability and Efficiency.} Figure \ref{fig:aggregates} and \ref{fig:packets} illustrate the average size of aggregated reports and the number of exchanged messages in WISE respectively (compared to the traditional approach of attesting the entire swarm). Both figures show how efficient and scalable WISE is. We do not show the total runtime of WISE as it is hardware-dependent. Please note that WISE is always faster than other existing swarm attestation schemes as it only targets a subset of devices. 


\textbf{Memory footprint and Computation Time.} At $\rho$ side, WISE requires less than 100 bytes of permanent storage for storing the secret keys and clusters IDs, whereas it requires no more than 5 MB at $\upsilon$ side to store a knowledge base ($<key, value>$ map) for a network consisting of 10k devices regardless the length of the attestation history. Computationally, for a network of 10k devices, the initial step that occurs once at $\upsilon$ side (computing the shortest paths, etc.) requires no more than 6 minutes. Then, the intelligence part of WISE requires less than 9 seconds each time. The verification of the attestation aggregates of 10K devices and updating the records in the knowledge base consumes no more than 3.5 seconds. The size of the attestation request did not exceed (at worst) 304 bytes for a network of 10k devices. 

\vspace{-5pt}\section{Conclusion}
\label{sec:conclusion}

This paper presented WISE, the first smart swarm attestation protocol for the Internet of Things that relies on a resource-efficient machine learning technique using Hidden Markov Model. WISE is scalable and highly efficient to be used in (quasi-) static networks consisting of thousands of heterogenous IoT devices. The minimized communication overhead  depends on the tolerated latency of detecting (unlikely) compromised devices. This makes WISE very suitable to be used in various application domains, especially, in time-sensitive and safety-critical ones such as smart factories.


\section*{Acknowledgment}

This research is supported by the research fund of KU Leuven and IMEC, a research institute founded by the Flemish government.

\bibliographystyle{ieeetr}

\bibliography{ref}

\end{document}